\documentclass[journal]{IEEEtran}
\usepackage[utf8]{inputenc}
\usepackage{amsmath}
\usepackage{soul}
\usepackage{array}
\usepackage{amsmath,amsthm,amssymb,epsf}
\usepackage{bm} 
\usepackage{graphicx,psfrag}
\usepackage{epstopdf}
\usepackage{verbatim}
\usepackage{multirow}
\usepackage{color}
\usepackage{setspace}
\usepackage{enumitem}
\usepackage{subfigure}
\usepackage{amssymb}
\usepackage{tabularx}
\usepackage{booktabs}
\usepackage{dirtytalk}
\usepackage{amsmath}
{
      \theoremstyle{plain}
      \newtheorem{assumption}{Assumption}
  }
\usepackage{csquotes}
\usepackage{algorithm}
\usepackage{algorithmic}

\newtheorem{theorem}{Theorem}
\newtheorem{definition}{Definition}

\newtheorem{remark}{Remark}

\usepackage{soul}
\usepackage{cite}

\begin{document}

\title{Extended Neighboring Extremal Optimal Control with State and Preview Perturbations}

\author{Amin~Vahidi-Moghaddam, Kaixiang~Zhang, Zhaojian~Li$^*$, Xunyuan Yin, Ziyou Song, and Yan~Wang
\thanks{$^*$Zhaojian Li is the corresponding author.}
\thanks{Amin Vahidi-Moghaddam, Kaixiang~Zhang, and Zhaojian Li are with the Department of Mechanical Engineering, Michigan State University, East Lansing, MI 48824 USA (e-mail: vahidimo@msu.edu,  zhangk64@msu.edu, lizhaoj1@egr.msu.edu).}
\thanks{Xunyuan Yin is with the School of Chemistry, Chemical Engineering and Biotechnology, Nanyang Technological University, 637459 Singapore (e-mail: xunyuan.yin@ntu.edu.sg).}
\thanks{Ziyou Song is with the Department of Mechanical Engineering,
National University of Singapore, 117575 Singapore (email:
ziyou@nus.edu.sg).}
\thanks{Yan Wang is with the Research and Advanced Engineering, Ford Motor Company, Dearborn, MI 48121 USA (e-mail: ywang21@ford.com).}
} 

\maketitle

\begin{abstract}
Optimal control schemes have achieved remarkable performance in numerous engineering applications. However, they typically require high computational cost, which has limited their use in real-world engineering systems with fast dynamics and/or limited computation power. To address this challenge, Neighboring Extremal (NE) has been developed as an efficient optimal adaption strategy to adapt a pre-computed nominal control solution to perturbations from the nominal trajectory. The resulting control law is a time-varying feedback gain that can be pre-computed along with the original optimal control problem, and it takes negligible online computation. However, existing NE frameworks only deal with state perturbations while in modern applications, optimal controllers (e.g., predictive controllers) frequently incorporate preview information. Therefore, a new NE framework is needed to adapt to such preview perturbations. In this work, an extended NE (ENE) framework is developed to systematically adapt the nominal control to both state and preview perturbations. We show that the derived ENE law is two time-varying feedback gains on the state perturbation and the preview perturbation. We also develop schemes to handle nominal non-optimal solutions and large perturbations to retain optimal performance and constraint satisfaction. Case study on nonlinear model predictive control is presented due to its popularity but it can be easily extended to other optimal control schemes. Promising simulation results on the cart inverted pendulum problem demonstrate the efficacy of the ENE algorithm.
\end{abstract}

\textbf{Note to Practitioners}. Due to the vast success in predictive control and advancement in sensing, modern control applications have frequently been incorporating preview information in the control design. For example, the road profile preview obtained from vehicle crowdsourcing is exploited for simultaneous suspension control and energy harvesting, demonstrating a significant performance enhancement using the preview information despite noises in the preview \cite{Hajida1}. Another example is thermal management for cabin and battery of hybrid electric vehicles, where traffic preview is employed in hierarchical model predictive control to improve energy efficiency \cite{amini2019cabin}. In \cite{laks2011model}, light detection and ranging systems are used to provide wind disturbance preview to enhance the controls of turbine blades. In \cite{yazdandoost2022optimization}, virtual water preview is employed using integrated water resources management modelling to optimize agricultural patterns and control level of water in lakes. In this work, we develop an extended neighboring extremal framework that can adapt a nominal control law to state and preview perturbations simultaneously. This setup is widely applicable as in many applications, a nominal preview is available while the preview signal can also be measured or estimated online. 


\IEEEpeerreviewmaketitle

\section{Introduction}
\label{Sec1}
Optimal control approaches, such as model predictive control (MPC), can explicitly handle system constraints while achieving optimal closed-loop performance \cite{bratta2022governor,ameli2022hierarchical,bieker2019deep}. However, such controllers typically involve solving optimization problems at each step and are thus computationally expensive, especially for nonlinear systems with non-convex constraints. This has hindered their wider adoption in applications with fast dynamics and/or limited computation resources \cite{ghaemi2006computationally}. The main motivation of this work is to address the high computational cost of the optimal control approaches such that it tackles the limitations of existing frameworks for modern applications which incorporate preview information.

As such, several frameworks have been developed to improve the computational efficiency of the optimal controllers. One approach is to simplify the system dynamics with model-reduction techniques \cite{lore2021model,amiri2022fly,zhang2022dimension}. However, these techniques require a trade-off between system performance and computational complexity, and it is often still computationally expensive after the model reduction. Another sound approach is to use function approximators, where functions such as neural networks \cite{bao2022learning,li2019adaptive,krishnamoorthy2021adaptive}, Gaussian process regression \cite{arcari2020meta,foumani2023multi,arcari2023bayesian}, and spatial temporal filters \cite{vahidi2022data} are exploited to learn the control policy, after which the learned policy is employed online to achieve efficient onboard computations. However, extensive data collection is required to ensure a comprehensive coverage of operating conditions. Furthermore, the learned policy lacks interpretability, and it is generally challenging to retain guaranteed system constraint satisfaction.

Neighboring extremal (NE) \cite{bagherzadeh2023neighboring,ghaemi2019optimal,gupta2017combined,bloch2016neighboring} is another promising paradigm to attain efficient computations by proposing a time-varying feedback gain on the state deviations. Specifically, given a pre-computed nominal solution based on a nominal initial state, the NE yields an optimal adaptation law (to the first order) that adapts the control to deviations from the nominal state, incurring negligible online computation while achieving (sub-)optimal performance when perturbations occur. The nominal solution can be computed offline and stored online, can be performed on a remote powerful controller (e.g., cloud), or computed ahead of time by utilizing the idling time of the processor. The NE framework has been employed in several engineering systems, including ship maneuvering control \cite{ghaemi2010path},  power management \cite{park2015real}, full bridge DC/DC converter \cite{xie2011model}, and spacecraft relative motion maneuvers \cite{park2013model}. In \cite{ghaemi2006computationally}, disturbance perturbations have been considered for the NE in the nonlinear optimal control problems; however, the formulation derived is limited to a constant disturbance. Moreover, using a parameter estimation for the systems with unknown parameters, parameter perturbations are considered for the NE, where the estimated parameters are considered constant during the predictions of the optimal control problem \cite{park2014tutorial}.

In this work, we aim to develop an extended NE (ENE) framework for the nonlinear optimal control problems to efficiently adapt a pre-computed nominal solution to both state and preview perturbations. The contributions of this work are: First, we formulate the ENE framework so that it surpasses the existing NE frameworks to explicitly handle both state and preview perturbations. This is a generalization of existing work \cite{ghaemi2008neighboring}, where only state perturbation is considered. Moreover, we treat the ENE problem when nominal non-optimal solution and large perturbations are present, and a multi-segment strategy is employed to guarantee constraint satisfaction in the presence of large perturbations. Furthermore, promising results are demonstrated by applying the developed control strategy to the cart inverted pendulum problem. Compared to our conference paper \cite{vahidi2022event}, we extend the NE framework so that we do not need to return the NMPC at several time steps to handle the time-varying preview perturbations, which significantly reduces the computational cost. These contributions are important extensions as in modern applications, optimal controllers frequently incorporate preview information (e.g., from a preview prediction model \cite{ahmadi2022deep}), and it is critical to adapt to the preview deviations to retain good system performance.

The paper is outlined as: Section~II describes the problem formulation and the preliminaries of the nonlinear optimal control problems. The proposed ENE framework is presented in Section~III. Simulation on the cart-inverted pendulum is presented in Section~IV. Finally, Section~V discusses conclusions and future works.

\section{Problem Formulation}
\label{Sec2}
In this section, preliminaries on nonlinear optimal control problems are reviewed, and perturbation analysis problem on the optimal solution is presented for the nonlinear systems with state and preview perturbations. Specifically, the following discrete-time nonlinear system, that incorporates a system preview, is considered as:
\begin{equation}
  \begin{aligned}
    \label{system}
    & x(k+1) = f(x(k),u(k),w(k)),
  \end{aligned}
\end{equation}
where $k \in \mathbb{N}^+$ represents the time step, $x \in \mathbb{R}^n$ denotes the measurable/observable states, and $u \in {\mathbb{R}^m}$ is the control input. Here $w \in \mathbb{R}^n$ represents the preview information that can be road profile preview in suspension controls \cite{Hajida1}, wind preview for turbine controls \cite{laks2011model}, and traffic preview in vehicle power management \cite{amini2019cabin}. Furthermore,
$f:\mathbb{R}^{n}\times \mathbb{R}^{m}\times \mathbb{R}^{n} \rightarrow \mathbb{R}^n$ represents the system dynamics with $f(0,0,0)=0$. Moreover, we consider the following general nominal preview model:
\begin{equation}
  \begin{aligned}
    \label{preview}
    & w(k+1) = g(x(k),w(k)),
  \end{aligned}
\end{equation}
where $g:\mathbb{R}^n\times \mathbb{R}^n \rightarrow \mathbb{R}^n$ represents the nominal preview dynamics. 

We consider the following safety constraints for the system: 
\begin{equation}
    \label{safety}
  C(x(k),u(k),w(k)) \leq 0,
\end{equation}
where $C:\mathbb{R}^n\times \mathbb{R}^m\times \mathbb{R}^n \rightarrow \mathbb{R}^l$.

\begin{definition}[Closed-Loop Performance]
\label{def1}
Consider the nonlinear system \eqref{system} and the control objective of regulating the state $x$. Starting from the initial conditions $x_0$ and $w_0$, the closed-loop system performance over $N$ steps is characterized by the following cost function:
\begin{equation}
    \label{cost}
    \begin{aligned}
  J_N(\mathbf{x},\mathbf{u},\mathbf{w})= \sum^{N-1}_{k=0} \phi(x(k),u(k),w(k)) + \psi(x(N),w(N)), 
  \end{aligned}
\end{equation}
where $\mathbf{x} = \left[ x(0),\, x(1),\, \cdots,\, x(N)  \right]$, $\mathbf{u} = \left[ u(0),\, u(1),\, \cdots,\, u(N-1)  \right]$, $\mathbf{w} = \left[ w(0),\, w(1),\, \cdots,\, w(N)  \right]$, and $\phi(x,u,w)$ and $\psi(x,w)$ denote the stage and terminal costs, respectively.
\end{definition}

\begin{assumption}[Twice Differentiable Functions]
\label{ass1}
The functions $f$, $g$, $C$, $\phi$, and $\psi$ are twice continuously differentiable.
\end{assumption}

With the defined closed-loop performance metric, the control goal is to minimize the cost function \eqref{cost} while adhering to the constraints in \eqref{system}-\eqref{safety}. The optimal control aims at optimizing the system performance over $N$ future steps for the system \eqref{system} using the nominal preview model \eqref{preview}, which is expressed as the following constrained optimization problem:
\begin{equation}
  \begin{aligned}
    \label{NMPC}
    &(\mathbf{x}^{o},\mathbf{u}^{o}) = \underset{\mathbf{x},\mathbf{u}}{\arg\min} \hspace{1 mm} J_N(\mathbf{x},\mathbf{u},\mathbf{w})\\
    &\text{s.t.} \hspace{5 mm} x(k+1) = f(x(k),u(k),w(k)),\\
    & \hspace{10 mm} w(k+1) = g(x(k),w(k)),\\
    & \hspace{10 mm} C(x(k),u(k),w(k)) \leq 0,\\
    & \hspace{10 mm} x(0)=x_0, \hspace{2 mm}  w(0)=w_0.
  \end{aligned}
\end{equation}

Consider a nominal trajectory $\mathbf{x}^{o}$, $\mathbf{u}^{o}$, and $\mathbf{w}^{o}$ obtained by solving (\ref{NMPC}) with $\mathbf{w}^{o}$ being the nominal preview. This computation can be performed on a remote powerful controller (e.g., cloud computing or edge computing) or can be computed ahead of time based on an approximated initial state. 
During implementation, the actual state $x(k)$ and the preview information $w(k)$ will likely deviate from the nominal trajectory. Let $\delta x(k) = x(k) - x^{o}(k)$ and $\delta w(k) = w(k) - w^{o}(k)$ denote the state perturbation and the preview perturbation, respectively. Now, to solve the nonlinear optimal control problem (\ref{NMPC}) for the actual values at each time step $k$, we seek a (sub-)optimal control update law, $u^{*}(k) = u^{o}(k) + \delta u(k)$, to efficiently adapt to the perturbations of the nominal trajectory. As such, using the nominal trajectory and the perturbation analysis, we develop an ENE framework to account for both state and preview perturbations through two time-varying feedback gains. Moreover, to handle large perturbations,
\begin{figure}[!ht]
     \centering
     \includegraphics[width=0.495\textwidth]{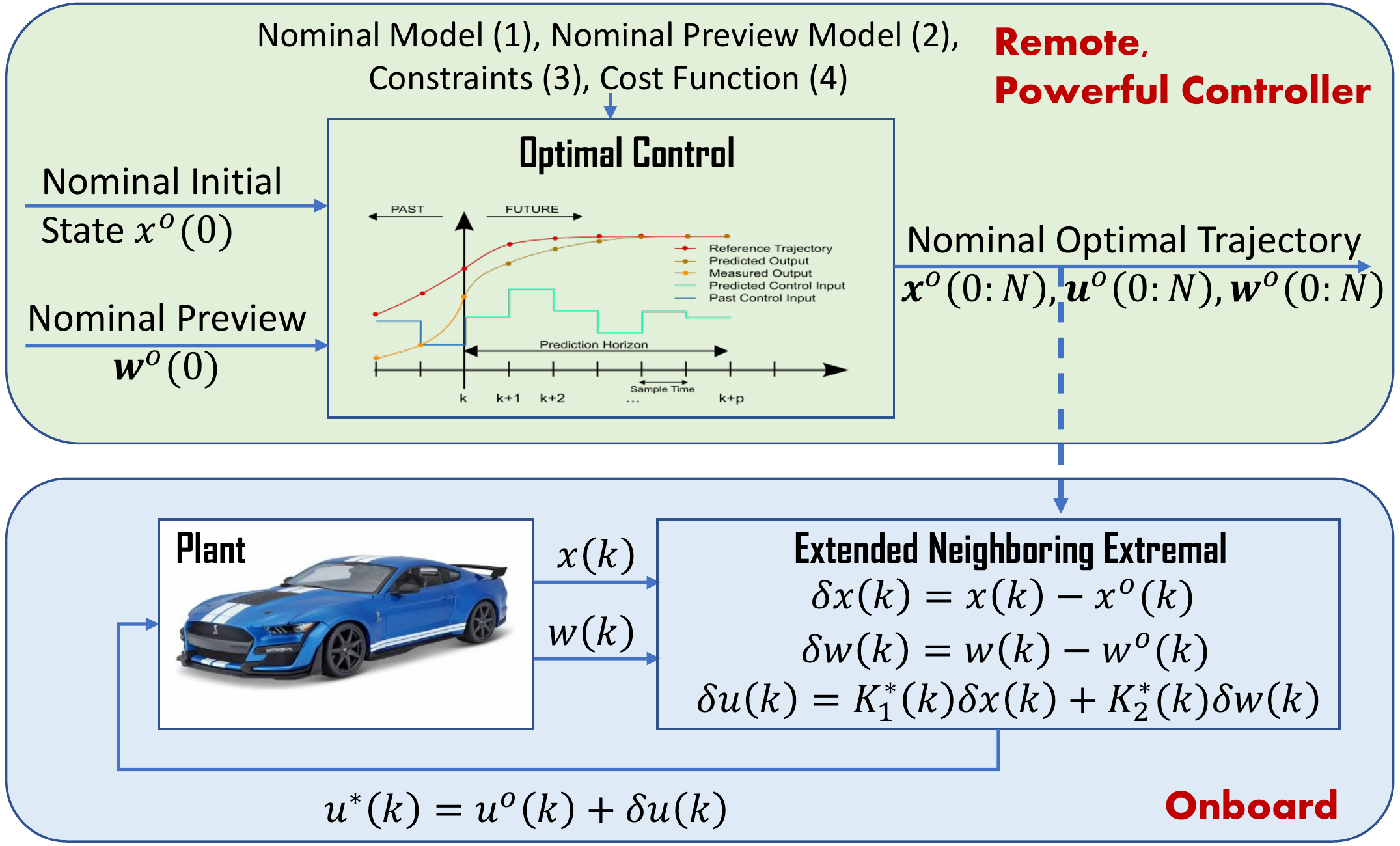}
    \caption{Schematic of the  extended neighboring extremal where the nominal solution is efficiently updated online based on state and preview perturbations.}
     \label{Fig0}
\end{figure}
we modify the ENE algorithm to preserve constraint satisfaction and retain optimal control performance. The details of each algorithm and their benefits for nonlinear model predictive control will be presented in the next part.

\section{Main Result}
In this section, we present an ENE framework for the optimal control problem \eqref{NMPC} subject to state and preview perturbations. As shown in Fig. 1, a nominal state and control trajectory is first computed based on system specifications (e.g., nominal model, nominal preview model, constraints, and cost function) along with a nominal initial state and preview. Then, the ENE approach exploits time-varying feedback gains to adapt to state and preview perturbations to retain optimal control performance. In the following subsections, we first analyze the nominal optimal solution to \eqref{NMPC} and perform the perturbation analysis to obtain an efficient optimal feedback law for small perturbations. We then develop schemes to handle large perturbations and maintain well control performance and constraint satisfaction.

\subsection{Nominal Optimal Solution}
In this subsection, we analyze the nominal optimal solutions using the Karush-Kuhn-Tucker (KKT) conditions. Specifically, define $\mathbb{K}^a$ and $\mathbb{K}^i$ as the sets of time steps at which the constrains are active (i.e., $C(x(k),u(k),w(k)) = 0$ in \eqref{safety}) and inactive (i.e., $C(x(k),u(k),w(k)) < 0$), respectively. From \eqref{NMPC}, the Hamiltonian function and the augmented cost function are defined as:
\begin{equation}
\small
  \begin{aligned}
    \label{Hamiltonian}
    &H(k) = \phi(x(k),u(k),w(k)) + \lambda^{T}(k+1) f(x(k),u(k),w(k)) \\ 
    & \hspace{11 mm} + \bar{\lambda}^{T}(k+1) g(x(k),w(k)) + \mu^{T}(k) C^{a}(x(k),u(k),w(k)),
  \end{aligned}
\end{equation}
\begin{equation}
\small
  \begin{aligned}
    \label{aug-cost}
  &\bar{J}_N(k) = \sum^{N-1}_{k=0} (H(k) - \lambda^{T}(k+1) x(k+1) - \bar{\lambda}^{T}(k+1) w(k+1)) \\ 
  & \hspace{13 mm} + \psi(x(N),w(N)),
    \end{aligned}
\end{equation}
where $C^{a}(x(k),u(k),w(k))$ represents the active constraints at the time step $k$. It is worth noting that $C^{a}(x(k),u(k),w(k))$ is an empty vector for inactive constraints, and $C^{a}(x(k),u(k),w(k)) \in \mathbb{R}^{l^a}$ if we have $l^a$ (out of $l$) active constraints. Furthermore, $\mu(k) \in \mathbb{R}^{l^a}$ is the Lagrange multiplier for the active constraints, and $\lambda(k+1) \in \mathbb{R}^n$ and $\bar{\lambda}(k+1) \in \mathbb{R}^n$ represent the Lagrange multipliers for the system dynamics \eqref{system} and the nominal preview model \eqref{preview}, respectively. It is worth noting that the Lagrange multipliers $\mu(k)$, $\lambda(k+1)$, and $\bar{\lambda}(k+1)$ are also referred as the co-states.

\begin{assumption}[Active Constraints]
\label{ass2}
At each time step $k$, the number of active constraints is not greater than $m$, i.e., $C^a_{u}(k)$ is full row rank.
\end{assumption}

Since $x^{o}(k)$, $u^{o}(k)$, and $w^{o}(k)$ ($k\in [0,N]$) represent the nominal optimal solution for the nonlinear optimal control problem \eqref{NMPC}, they satisfy the following KKT conditions for the augmented cost function \eqref{aug-cost}:

\begin{equation}
  \begin{aligned}
    \label{KKT}
    &H_{u}(k) = 0, \hspace{1 mm} k = 0, 1, ..., N-1, \\
    &\lambda(k) = H_{x}(k), \hspace{1 mm} k = 0, 1, ..., N-1, \\
    &\lambda(N) = \psi_{x}(x(N),w(N)), \\
    &\bar{\lambda}(k) = H_{w}(k), \hspace{1 mm} k = 0, 1, ..., N-1, \\
    &\bar{\lambda}(N) = \psi_{w}(x(N),w(N)), \\
    &\mu(k) \geq 0, \hspace{1 mm} k = 0, 1, ..., N-1,
  \end{aligned}
\end{equation}
where the subscripts $u$, $x$, and $w$ represent the partial derivatives of a function. 

Now, using the KKT conditions and the nominal solution $x^{o}(k)$, $u^{o}(k)$, and $w^{o}(k)$, one can calculate the Lagrange multipliers $\mu(k)$, $\lambda(k+1)$, and $\bar{\lambda}(k+1)$ online, that is, from (\ref{KKT}), it follows that
\begin{equation}
  \begin{aligned}
    \label{Optimal KKT}
    &\hspace{2.5 mm} 0 \hspace{3.75 mm} = \phi_{u}(x^{o},u^{o},w^{o}) + \lambda^{T}(k+1) f_{u}(x^{o},u^{o},w^{o})\\
    & \hspace{11.5 mm} + \mu^{T}(k) C^a_{u}(x^{o},u^{o},w^{o}),\\
    & \hspace{0.5 mm} \lambda(k) \hspace{0.75 mm} = \phi_{x}(x^{o},u^{o},w^{o}) + \lambda^{T}(k+1) f_{x}(x^{o},u^{o},w^{o}) \\
    & \hspace{11.5 mm} + \bar{\lambda}^{T}(k+1) g_{x}(x^{o},w^{o}) + \mu^{T}(k) C^a_{x}(x^{o},u^{o},w^{o}),\\
    &\lambda(N) = \psi_{x}(x^{o}(N),w^{o}(N)),\\
    & \hspace{0.5 mm} \bar{\lambda}(k) \hspace{0.75 mm} = \phi_{w}(x^{o},u^{o},w^{o}) + \lambda^{T}(k+1) f_{w}(x^{o},u^{o},w^{o}) \\
    & \hspace{11.5 mm} + \bar{\lambda}^{T}(k+1) g_{w}(x^{o},w^{o}) + \mu^{T}(k) C^a_{w}(x^{o},u^{o},w^{o}),\\
    &\bar{\lambda}(N) = \psi_{w}(x^{o}(N),w^{o}(N)).
  \end{aligned}
\end{equation}

Using the above equations, the Lagrange multipliers can be obtained as:
\begin{equation}
  \begin{aligned}
    \label{Lagrange multipliers}
    &\mu(k) = -{(C^a_{u}(k) {C^a_{u}}^T(k))}^{-1} C^a_{u}(k) \phi^T_{u}(k)\\ 
    & \hspace{10.5 mm} -{(C^a_{u}(k) {C^a_{u}}^T(k))}^{-1} C^a_{u}(k) f^T_{u}(k) \lambda(k+1),\\
    &\lambda(k) = \phi_{x}(k) + \lambda^{T}(k+1) f_{x}(k) + \bar{\lambda}^{T}(k+1) g_{x}(k)\\
    & \hspace{10.5 mm}  + \mu^{T}(k) C^a_{x}(k),\\
    &\bar{\lambda}(k) = \phi_{w}(k) + \lambda^{T}(k+1) f_{w}(k) + \bar{\lambda}^{T}(k+1) g_{w}(k)\\
    & \hspace{10.5 mm}  + \mu^{T}(k) C^a_{w}(k).
  \end{aligned}
\end{equation}

Note that Assumption \ref{ass2} guarantees that $C^a_{u}(k) {C^a_{u}}^T(k)$ is invertible,
and $\delta \bar{J}_N(x^{o},u^{o},w^{o},\mu^{o}, \lambda^{o}, \bar{\lambda}^{o}) = 0$.

\subsection{Extended Neighboring Extremal}
\label{Sec3}
For this part, we assume that the state and preview perturbations are small enough such that they do not change the activity status of the constraint. To adapt to state and preview perturbations from the nominal values, the ENE seeks to minimize the second-order variation of \eqref{aug-cost} subject to linearized models and constraints. More specifically, the ENE algorithm solves the following optimization problem with the initial conditions $\delta x(0)$ and $\delta w(0)$ as: 
\begin{equation}
  \begin{aligned}
    \label{NE}
    &\mathbf{\delta u^{*}} = \underset{\mathbf{\delta u}}{\arg\min} \hspace{1 mm} {J}^{ne}_N(k)\\
  &\text{s.t.} \hspace{5 mm} \delta x(k+1) = f_{x}(k) \delta x(k) + f_{u}(k) \delta u(k)+f_{w}(k) \delta w(k),\\ 
  & \hspace{10 mm} \delta w(k+1) = g_{x}(k) \delta x(k) + g_{w}(k) \delta w(k),\\
  & \hspace{10 mm} C^{a}_{x}(k) \delta x(k) + C^{a}_{u}(k) \delta u(k) + C^{a}_{w}(k) \delta w(k) = 0,
  \\
  & \hspace{10 mm} \delta x(0) = \delta x_{0}, \hspace{2 mm} \delta w(0) = \delta w_{0},
  \end{aligned}
\end{equation}
where \begin{equation}
  \begin{aligned}
    \label{NE-cost}
  &{J}^{ne}_N(k)=\delta^{2} \bar{J}_N(k) = \\
  &\frac{1}{2} \sum^{N-1}_{k=0} 
  \begin{bmatrix}
  \delta x(k)\\ \delta u(k)\\ \delta w(k)
  \end{bmatrix}^{T}
  \begin{bmatrix}
  H_{xx}(k) & H_{xu}(k) & H_{xw}(k)\\
  H_{ux}(k) & H_{uu}(k) & H_{uw}(k)\\
  H_{wx}(k) & H_{wu}(k) & H_{ww}(k)
  \end{bmatrix}
  \begin{bmatrix}
  \delta x(k)\\ \delta u(k)\\ \delta w(k)
  \end{bmatrix} \\
  &+ \frac{1}{2} \delta x^{T} (N) \psi_{xx}(N) \delta x (N) + \frac{1}{2} \delta w^{T} (N) \psi_{ww}(N) \delta w (N)
    \end{aligned}
\end{equation}

For \eqref{NE} and \eqref{NE-cost}, the Hamiltonian function and the augmented cost function are obtained as
\begin{equation}
  \begin{aligned}
    \label{NE-Hamiltonian}
    &H^{ne}(k) = \\
    &\frac{1}{2} 
    \begin{bmatrix}
    \delta x(k)\\ \delta u(k)\\ \delta w(k)
    \end{bmatrix}^{T}
    \begin{bmatrix}
    H_{xx}(k) & H_{xu}(k) & H_{xw}(k)\\
    H_{ux}(k) & H_{uu}(k) & H_{uw}(k)\\
    H_{wx}(k) & H_{wu}(k) & H_{ww}(k)
    \end{bmatrix}
    \begin{bmatrix}
    \delta x(k)\\ \delta u(k)\\ \delta w(k)
    \end{bmatrix} \\
    &+ \delta \lambda^{T}(k+1) (f_{x}(k) \delta x(k) + f_{u}(k) \delta u(k) + f_{w}(k) \delta w(k))\\
    &+ \delta \bar{\lambda}^{T}(k+1) (g_{x}(k) \delta x(k) + g_{w}(k) \delta w(k))\\
    &+ \delta \mu^{T}(k) (C^{a}_{x}(k) \delta x(k) + C^{a}_{u}(k) \delta u(k) + C^{a}_{w}(k) \delta w(k)),
  \end{aligned}
\end{equation}
\begin{equation}
\small
  \begin{aligned}
    \label{NE-aug-cost}
  &\bar{J}^{ne}_N(k) = \\
  &\sum^{N-1}_{k=0} (H^{ne}(k) - \delta \lambda^{T}(k+1) \delta x(k+1)  - \delta \bar{\lambda}^{T}(k+1) \delta w(k+1)) \\ 
  &+ \frac{1}{2} \delta {x}^{T} (N) \psi_{xx}(N) \delta x (N) + \frac{1}{2} \delta {w}^{T} (N) \psi_{ww}(N) \delta w (N),
    \end{aligned}
\end{equation}
where $\delta \mu(k)$, $\delta \lambda(k)$, and $\delta \bar{\lambda}(k)$ are the Lagrange multipliers.
By applying the KKT conditions to \eqref{NE-aug-cost}, one has
\begin{equation}
  \begin{aligned}
    \label{NE-KKT}
    &H^{ne}_{\delta u}(k) = 0, \hspace{1 mm} k = 0, 1, ..., N-1, \\
    &\delta \lambda(k) = H^{ne}_{\delta x}(k), \hspace{1 mm} k = 0, 1, ..., N-1, \\
    &\delta \lambda(N) = \psi_{xx}(N) \delta x(N), \\
    &\delta \bar{\lambda}(k) = H^{ne}_{\delta w}(k), \hspace{1 mm} k = 0, 1, ..., N-1, \\
    &\delta \bar{\lambda}(N) = \psi_{ww}(N) \delta w(N), \\
    &\delta \mu(k) \geq 0, \hspace{1 mm} k = 0, 1, ..., N-1.
      \end{aligned}
\end{equation}

To facilitate the development of the ENE algorithm, several auxiliary variables are introduced for the following theorem. Specifically, for $k=1,2, \cdots, N-1$, $S(k)$, $W(k)$, $\bar{S}(k)$, and $\bar{W}(k)$ are defined as 
\begin{equation}
  \begin{aligned}
    \label{S_k}
    &S(k) = Z_{xx}(k) -
    \begin{bmatrix}
    Z_{xu}(k) \hspace{1 mm} {C^a_{x}}^T(k)
    \end{bmatrix}
    K^o(k)
    \begin{bmatrix}
    Z_{ux}(k)\\
    C^a_{x}(k)
    \end{bmatrix}, \\
    &W(k) = Z_{x w}(k) -
    \begin{bmatrix}
    Z_{xu}(k) \hspace{1 mm} {C^a_{x}}^T(k)
    \end{bmatrix}
    K^o(k)
    \begin{bmatrix}
    Z_{u w}(k)\\
    C^a_{w}(k)
    \end{bmatrix},
      \end{aligned}
\end{equation}
\vspace{-12pt}
\begin{equation}
  \begin{aligned}
    \label{S_k_bar}
    &\bar{S}(k) = Z_{wx}(k) -
    \begin{bmatrix}
    Z_{wu}(k) \hspace{1 mm} {C^a_{w}}^T(k)
    \end{bmatrix}
    K^o(k)
    \begin{bmatrix}
    Z_{ux}(k)\\
    C^a_{x}(k)
    \end{bmatrix}, \\
    &\bar{W}(k) = Z_{ww}(k) -
    \begin{bmatrix}
    Z_{wu}(k) \hspace{1 mm} {C^a_{w}}^T(k)
    \end{bmatrix}
    K^o(k)
    \begin{bmatrix}
    Z_{u w}(k)\\
    C^a_{w}(k)
    \end{bmatrix},
      \end{aligned}
\end{equation}
where
\begin{equation} 
\label{K0}
 K^{o}(k) =\begin{cases} 
    \begin{bmatrix}
    Z_{uu}(k) & {C^a_{u}}^T(k)\\
    C^a_{u}(k) & 0
    \end{bmatrix}^{-1} & \text{if} \hspace{2 mm} k \in \mathbb{K}^a, \vspace{2 mm}\\
    \begin{bmatrix}
    Z^{-1}_{uu}(k) & 0\\
    0 & 0
    \end{bmatrix} &  \text{if} \hspace{2 mm} k \in \mathbb{K}^i,
    \end{cases}
\end{equation}
and
\begin{equation}
  \begin{aligned}
    \label{Zuu}
    &Z_{ux}(k) = H_{ux}(k) + f^{T}_{u}(k) S(k+1) f_{x}(k)\\
    & \hspace{14.5 mm}+ f^{T}_{u}(k) W(k+1) g_{x}(k),\\
    &Z_{uu}(k) = H_{uu}(k) + f^{T}_{u}(k) S(k+1) f_{u}(k),\\
    &Z_{u w}(k) = H_{u w}(k) + f^{T}_{u}(k) S(k+1) f_{w}(k)\\
    & \hspace{14.5 mm} + f^{T}_{u}(k) W(k+1) g_{w}(k),\\
    &Z_{xx}(k) = H_{xx}(k) + f^{T}_{x}(k) S(k+1) f_{x}(k)\\
    & \hspace{14.5 mm} +f^{T}_{x}(k) W(k+1) g_{x}(k) + g^{T}_{x}(k) \bar{S}(k+1) f_{x}(k) \\
    & \hspace{14.5 mm}+ g^{T}_{x}(k) \bar{W}(k+1) g_{x}(k),\\
    &Z_{xu}(k) = H_{xu}(k) + f^{T}_{x}(k) S(k+1) f_{u}(k)\\
    & \hspace{14.5 mm} + g^{T}_{x}(k) \bar{S}(k+1) f_{u}(k),\\
    &Z_{x w}(k) = H_{x w}(k) + f^{T}_{x}(k) S(k+1) f_{w}(k)\\
    & \hspace{14.5 mm}+ f^{T}_{x}(k) W(k+1) g_{w}(k)+ g^{T}_{x}(k) \bar{S}(k+1) f_{w}(k)\\
    & \hspace{14.5 mm} + g^{T}_{x}(k) \bar{W}(k+1) g_{w}(k),\\
    &Z_{wx}(k) = H_{wx}(k) + f^{T}_{w}(k) S(k+1) f_{x}(k)\\
    & \hspace{14.5 mm} + f^{T}_{w}(k) W(k+1) g_{x}(k)+ g^{T}_{w}(k) \bar{S}(k+1) f_{x}(k)\\
    & \hspace{14.5 mm} + g^{T}_{w}(k) \bar{W}(k+1) g_{x}(k),\\
    &Z_{wu}(k) = H_{wu}(k) + f^{T}_{w}(k) S(k+1) f_{u}(k)\\
    & \hspace{14.5 mm} + g^{T}_{w}(k) \bar{S}(k+1) f_{u}(k),\\    
    &Z_{ww}(k) = H_{ww}(k) + f^{T}_{w}(k) S(k+1) f_{w}(k)\\
    & \hspace{14.5 mm} + f^{T}_{w}(k) W(k+1) g_{w}(k) + g^{T}_{w}(k) \bar{S}(k+1) f_{w}(k)\\
    & \hspace{14.5 mm} + g^{T}_{w}(k) \bar{W}(k+1) g_{w}(k).
      \end{aligned}
\end{equation}

The terminal conditions for $S(k)$, $W(k)$, $\bar{S}(k)$, and $\bar{W}(k)$ are given by
\begin{equation}
    \label{S_N}
    S(N) = \psi_{xx}(N), \qquad W(N) = 0,
\end{equation}
\begin{equation}
    \label{S_N_bar}
    \bar{S}(N) = 0, \qquad \bar{W}(N) = \psi_{ww}(N).
\end{equation}

\begin{theorem} [Extended Neighboring Extremal]
\label{theo1}
Consider the optimization problem \eqref{NE}, the Hamiltonian function \eqref{NE-Hamiltonian}, the KKT conditions \eqref{NE-KKT}, and the defined auxiliary variables \eqref{S_k} and \eqref{S_k_bar}. If $Z_{uu}(k)>0$ for $k\in \begin{bmatrix}
0, N-1
\end{bmatrix}$, then the ENE policy
\begin{equation}
  \begin{aligned}
    \label{law}
    &\delta u(k) = K^{*}_1(k) \delta x(k) + K^{*}_2(k) \delta w(k),\\
    &K^{*}_1(k) = -
    \begin{bmatrix}
    I & 0
    \end{bmatrix}
    K^o(k)
    \begin{bmatrix}
    Z_{ux}(k)\\
    {C^a_{x}}(k)
    \end{bmatrix},\\
    &K^{*}_2(k) = -
    \begin{bmatrix}
    I & 0
    \end{bmatrix}
    K^o(k)
    \begin{bmatrix}
    Z_{u w}(k)\\
    {C^a_{w}}(k)
    \end{bmatrix},\\
  \end{aligned}
\end{equation}
approximates the perturbed solution for the nonlinear optimal control problem \eqref{NMPC} in the presence of state perturbation $\delta x(k)$ and preview perturbation $\delta w(k)$.
\end{theorem}

\begin{proof}
Using \eqref{NE-Hamiltonian}, \eqref{NE-aug-cost}, and the KKT conditions \eqref{NE-KKT}, one has 
\begin{equation}
  \begin{aligned}
    \label{H_u}
    &H_{ux}(k) \delta x(k) + H_{uu}(k) \delta u(k) + H_{u w}(k) \delta w(k) \\
    &+ f^T_{u}(k) \delta \lambda(k+1) + {C^a_{u}}^T(k) \delta \mu(k) = 0,
      \end{aligned}
\end{equation}
\begin{equation}
  \begin{aligned}
    \label{H_x}
    &\delta \lambda(k) = H_{xx}(k) \delta x(k) + H_{xu}(k) \delta u(k) + H_{x w}(k) \delta w(k)\\
    &+ f^T_{x}(k) \delta \lambda(k+1) + g^T_{x}(k) \delta \bar{\lambda}(k+1) + {C^a_{x}}^T(k) \delta \mu(k),
      \end{aligned}
\end{equation}
\begin{equation}
  \begin{aligned}
    \label{H_w}
    &\delta \bar{\lambda}(k) = H_{wx}(k) \delta x(k) + H_{wu}(k) \delta u(k) + H_{ww}(k) \delta w(k)\\
    &+ f^T_{w}(k) \delta \lambda(k+1) + g^T_{w}(k) \delta \bar{\lambda}(k+1) + {C^a_{w}}^T(k) \delta \mu(k),
      \end{aligned}
\end{equation}
where $\delta \lambda(N) = \psi_{xx}(N) \delta x(N)$ and $\delta \bar{\lambda}(N) = \psi_{ww}(N) \delta w(N)$. 
Now, define the following general relation:
\begin{equation}
  \begin{aligned}
    \label{lambda_form}
    &\delta \lambda(k) = S(k) \delta x(k) + W(k) \delta w(k) +T(k),
      \end{aligned}
\end{equation}
\begin{equation}
  \begin{aligned}
    \label{lambda_form_bar}
    &\delta \bar{\lambda}(k) = \bar{S}(k) \delta x(k) + \bar{W}(k) \delta w(k) +\bar{T}(k).
      \end{aligned}
\end{equation}
Using \eqref{NE-KKT}, 
\eqref{lambda_form}, and \eqref{lambda_form_bar}, one has $T(N) = 0$ and $\bar{T}(N) = 0$.
Substituting the linearized model \eqref{NE} and \eqref{lambda_form} into \eqref{H_u} yields
\begin{equation}
  \begin{aligned}
    \label{Z_uu}
    &Z_{ux}(k) \delta x(k) + Z_{uu}(k) \delta u(k) + Z_{u w}(k) \delta w(k) \\
    &+ {C^a_{u}}^T(k) \delta \mu(k) + f^T_{u}(k) T(k+1) = 0.
      \end{aligned}
\end{equation}
Using the linearized safety constraints \eqref{NE} and \eqref{Z_uu}, one has
\begin{equation}
  \begin{aligned}
    \label{u&mu}
    &\begin{bmatrix}
    \delta u(k)\\
    \delta \mu(k)
    \end{bmatrix}
    = -K^o(k)
    \begin{bmatrix}
    Z_{ux}(k) \\
    C^a_{x}(k)
    \end{bmatrix}
    \delta x(k) \\
    & \hspace{17 mm} -K^o(k)
    \begin{bmatrix}
    Z_{u w}(k) \\
    C^a_{w}(k)
    \end{bmatrix}
    \delta w(k) \\
    & \hspace{17 mm} -K^o(k)
    \begin{bmatrix}
    f^T_{u}(k) T(k+1) \\
    0
    \end{bmatrix}.
  \end{aligned}
\end{equation}
Now, substituting the model \eqref{NE}, \eqref{lambda_form} and \eqref{lambda_form_bar} into \eqref{H_x} yields
\begin{equation}
  \begin{aligned}
    \label{lambda1}
    &\delta \lambda(k) = Z_{xx}(k) \delta x(k) + Z_{xu}(k) \delta u(k) + Z_{x w}(k) \delta w(k)\\
    &\hspace{12 mm}+ {C^a_{x}}^T(k) \delta \mu(k) + f^T_{x}(k) T(k+1) + g^T_{x}(k) \bar{T}(k+1).
      \end{aligned}
\end{equation}
Furthermore, substituting \eqref{u&mu} into \eqref{lambda1} yields
\begin{equation}
\small
  \begin{aligned}
    \label{lambda2}
    &\delta \lambda(k) = \left( Z_{xx}(k) -
    \begin{bmatrix}
    Z_{xu}(k) \hspace{1 mm} {C^{a}_{x}}^T(k)
    \end{bmatrix}
    K^o(k)
    \begin{bmatrix}
    Z_{ux}(k)\\
    C^{a}_{x}(k)
    \end{bmatrix} \right) \delta x(k) \\
    &+ \left( Z_{x w}(k) -
    \begin{bmatrix}
    Z_{xu}(k) \hspace{1 mm} {C^{a}_{x}}^T(k)
    \end{bmatrix}
    K^o(k)
    \begin{bmatrix}
    Z_{u w}(k)\\
    C^a_{w}(k)
    \end{bmatrix} \right) \delta w(k) \\
    &+ f^T_{x}(k) T(k+1) -
    \begin{bmatrix}
    Z_{xu}(k) \hspace{1 mm} {C^{a}_{x}}^T(k)
    \end{bmatrix}
    K^o(k)
    \begin{bmatrix}
    f^T_{u}(k) T(k+1)\\
    0
    \end{bmatrix}\\
    &+ g^T_{x}(k) \bar{T}(k+1).
      \end{aligned}
\end{equation}
From \eqref{S_k}, \eqref{lambda_form} and \eqref{lambda2}, it can be concluded that
\begin{equation}
  \begin{aligned}
    \label{T_k}
    &T(k) = g^T_{x}(k) \bar{T}(k+1) + f^T_{x}(k) T(k+1) \\
    & \hspace{11 mm} - \begin{bmatrix}
    Z_{xu}(k) \hspace{1 mm} {C^{a}_{x}}^T(k)
    \end{bmatrix}
    K^o(k)
    \begin{bmatrix}
    f^T_{u}(k) T(k+1)\\
    0
    \end{bmatrix}.
      \end{aligned}
\end{equation}
Now, substituting the model \eqref{NE}, \eqref{lambda_form} and \eqref{lambda_form_bar} into \eqref{H_w} yields
\begin{equation}
  \begin{aligned}
    \label{lambda_bar1}
    &\delta \bar{\lambda}(k) = Z_{wx}(k) \delta x(k) + Z_{wu}(k) \delta u(k) + Z_{ww}(k) \delta w(k)\\
    &\hspace{12 mm}+ {C^a_{w}}^T(k) \delta \mu(k) + f^T_{w}(k) T(k+1) + g^T_{w}(k) \bar{T}(k+1).
      \end{aligned}
\end{equation}
Furthermore, plugging \eqref{u&mu} into \eqref{lambda_bar1} yields
\begin{equation}
\small
  \begin{aligned}
    \label{lambda_bar2}
    &\delta\bar{\lambda}(k) = \left( Z_{wx}(k) -
    \begin{bmatrix}
    Z_{wu}(k) \hspace{1 mm} {C^{a}_{w}}^T(k)
    \end{bmatrix}
    K^o(k)
    \begin{bmatrix}
    Z_{ux}(k)\\
    C^{a}_{x}(k)
    \end{bmatrix} \right) \delta x(k) \\
    &+ \left( Z_{w w}(k) -
    \begin{bmatrix}
    Z_{wu}(k) \hspace{1 mm} {C^{a}_{w}}^T(k)
    \end{bmatrix}
    K^o(k)
    \begin{bmatrix}
    Z_{u w}(k)\\
    C^a_{w}(k)
    \end{bmatrix} \right) \delta w(k) \\
    &+ f^T_{w}(k) T(k+1) -
    \begin{bmatrix}
    Z_{wu}(k) \hspace{1 mm} {C^{a}_{w}}^T(k)
    \end{bmatrix}
    K^o(k)
    \begin{bmatrix}
    f^T_{u}(k) T(k+1)\\
    0
    \end{bmatrix}\\
    &+ g^T_{w}(k) \bar{T}(k+1). 
      \end{aligned}
\end{equation}
Using \eqref{S_k_bar}, \eqref{lambda_form_bar} and \eqref{lambda_bar2}, one has
\begin{equation}
  \begin{aligned}
    \label{T_k_bar}
    &\bar{T}(k) = g^T_{w}(k) \bar{T}(k+1) + f^T_{w}(k) T(k+1) \\
    & \hspace{11 mm} - \begin{bmatrix}
    Z_{wu}(k) \hspace{1 mm} {C^{a}_{w}}^T(k)
    \end{bmatrix}
    K^o(k)
    \begin{bmatrix}
    f^T_{u}(k) T(k+1)\\
    0
    \end{bmatrix}.
      \end{aligned}
\end{equation}
Based on \eqref{T_k}, \eqref{T_k_bar}, and the fact that $T(N)=0$, $\bar{T}(N)=0$, one can conclude that for $k\in \begin{bmatrix}
1, N-1
\end{bmatrix}$, $T(k)=0$, $\bar{T}(k)=0$. 
Thus, by using \eqref{u&mu}, the ENE policy \eqref{law} can be obtained. This completes the proof.
\end{proof}

\begin{remark} [Singularity]
\label{Singularity}
It is worth noting that the assumption of $Z_{uu}$ being positive definite (i.e., $Z_{uu}(k)>0, k \in \begin{bmatrix} 0, N-1 \end{bmatrix}$) is essential for the ENE. $Z_{uu}(k)>0$ is performed to calculate the ENE such that it guarantees the convexity of \eqref{NE}. Considering $Z_{uu}(k)>0$ and Assumption \ref{ass2}, it is clear that $K^o(k)$ in \eqref{K0} is well defined. However, when the constraints involve only state and preview (i.e., $C^a_{u}(k)=0$), or when $l^a$ is greater than $m$ (i.e., $C^a_{u}(k)$ is not full row rank), the matrix $K^o$ is singular, leading to the failure of the proposed algorithm. This issue can be solved using the constraint back-propagation algorithm presented in \cite{ghaemi2008neighboring}.
\end{remark}

\begin{remark} [Nominal Preview Model]
\label{Nominal Preview Model}
If we do not have any idea about the nominal preview model \eqref{preview} for the existing preview information in the real system, we can simply use $w(k+1)=w(k)$ as the nominal preview model for the nonlinear optimal control problem \eqref{NMPC} and the ENE algorithm. However, it is clear that we achieve the best performance using the ENE when the nominal preview model describes the preview information perfectly.
\end{remark}

Algorithm 1 summarizes the ENE procedure for adaptation the pre-computed nominal control solution $u^{o}(k)$ to the small state perturbation $\delta x(k)$ and the small preview perturbation $\delta w(k)$ such that it achieves the optimal control as $u^{*}(k) = u^{o}(k) + \delta u(k)$ using Theorem \ref{theo1}.

\begin{algorithm}[ht]
    \caption{Extended Neighboring Extremal}
    \label{ENE}
    \textbf{Input}: The functions $f$, $g$, $C$, $\phi$, and $\psi$, and the nominal trajectory $\mathbf{x}^{o}(0:N)$, $\mathbf{u}^{o}(0:N)$, and $\mathbf{w}^{o}(0:N)$.\\
    \textbf{1}: Initialize the matrices $\lambda^o (N)$, $\bar{\lambda}^o(N)$, $S(N)$, $W(N)$, $\bar{S}(N)$, and $\bar{W}(N)$ using \eqref{Optimal KKT}, \eqref{S_N} , and \eqref{S_N_bar}, respectively.\\
    \textbf{2}: Calculate, in a backward run, the Lagrange multipliers $\mu^o(k)$, $\lambda^o(k)$, and $\bar{\lambda}^o(k)$ using \eqref{Lagrange multipliers}.\\ 
    \textbf{3}: Calculate, in a backward run, the matrices $Z(k)$, the gains $K^*_1(k)$ and $K^*_2(k)$, and the matrices $S(k)$, $W(k)$, $\bar{S}(k)$, and $\bar{W}(k)$ using \eqref{Zuu}, \eqref{law}, \eqref{S_k}, and \eqref{S_k_bar}, respectively.\\
    \textbf{4}: Given $x^o(0)$, $w^o(0)$, $\delta x(0)$, and $\delta w(0)$, in a forward run, calculate $\delta u(k)$, $u^*(k)$, $x(k+1)$, and $w(k+1)$ using \eqref{law} and \eqref{system}.
\end{algorithm}

\subsection{Nominal Non-Optimal Solution and Large Perturbations}
The ENE is derived under the assumption that a nominal optimal solution is available, and the state and preview perturbations are small such that they do not change the activity status of the constraints. In this subsection, we modify the ENE policy for a nominal non-optimal solution and accordingly improve the algorithm to handle large state and preview perturbations which may change the sets of inactive and active constraints.

For the nominal non-optimal sequences $x^{o}(k)$, $u^{o}(k)$, $w^{o}(k)$, $\mu^{o}(k)$, $\lambda^{o}(k)$, and $\bar{\lambda}^{o}(k)$, we assume that they satisfy the constraints described in \eqref{NMPC} and \eqref{KKT} but may not satisfy the optimality condition $H_{u}(x^{o},u^{o},w^{o},\mu^{o}, \lambda^{o}, \bar{\lambda}^{o}) = 0$. Under this circumstance, the cost function \eqref{NE-cost} is modified as 
\begin{equation}
  \begin{aligned}
    \label{NE-cost-large}
  &{J}^{ne}_N(k)=\delta^{2} \bar{J}_N(k) + \sum^{N-1}_{k=0} {H}^{T}_u(k) \delta u(k)= \\
  &\frac{1}{2} \sum^{N-1}_{k=0} 
  \begin{bmatrix}
  \delta x(k)\\ \delta u(k)\\ \delta w(k)
  \end{bmatrix}^{T}
  \begin{bmatrix}
  H_{xx}(k) & H_{xu}(k) & H_{xw}(k)\\
  H_{ux}(k) & H_{uu}(k) & H_{uw}(k)\\
  H_{wx}(k) & H_{wu}(k) & H_{ww}(k)
  \end{bmatrix}
  \begin{bmatrix}
  \delta x(k)\\ \delta u(k)\\ \delta w(k)
  \end{bmatrix} \\
  &+ \frac{1}{2} \delta x^{T} (N) \psi_{xx}(N) \delta x (N) + \frac{1}{2} \delta w^{T} (N) \psi_{ww}(N) \delta w (N)\\
  & + \sum^{N-1}_{k=0} {H}^{T}_u(k) \delta u(k).
    \end{aligned}
\end{equation}

Considering the optimal control problem \eqref{NE} and the cost function \eqref{NE-cost-large}, the Hamiltonian function is modified as
\begin{equation}
  \begin{aligned}
    \label{NE-Hamiltonian-large}
    &H^{ne}(k) = \\
    &\frac{1}{2} 
    \begin{bmatrix}
    \delta x(k)\\ \delta u(k)\\ \delta w(k)
    \end{bmatrix}^{T}
    \begin{bmatrix}
    H_{xx}(k) & H_{xu}(k) & H_{xw}(k)\\
    H_{ux}(k) & H_{uu}(k) & H_{uw}(k)\\
    H_{wx}(k) & H_{wu}(k) & H_{ww}(k)
    \end{bmatrix}
    \begin{bmatrix}
    \delta x(k)\\ \delta u(k)\\ \delta w(k)
    \end{bmatrix} \\
    &+ \delta \lambda^{T}(k+1) (f_{x}(k) \delta x(k) + f_{u}(k) \delta u(k) + f_{w}(k) \delta w(k))\\
    &+ \delta \bar{\lambda}^{T}(k+1) (g_{x}(k) \delta x(k) + g_{w}(k) \delta w(k))\\
    &+ \delta \mu^{T}(k) (C^{a}_{x}(k) \delta x(k) + C^{a}_{u}(k) \delta u(k) + C^{a}_{w}(k) \delta w(k))\\
    &+{H}^{T}_u(k) \delta u(k).
  \end{aligned}
\end{equation}

Now, the following theorem is presented to modify the ENE policy for the nominal non-optimal solutions to the nonlinear optimal control problem \eqref{NMPC}.

\begin{theorem} [Modified Extended Neighboring Extremal]
\label{theo2}
Consider the optimization problem \eqref{NE}, the KKT conditions \eqref{NE-KKT}, and the Hamiltonian function \eqref{NE-Hamiltonian-large}. If $Z_{uu}(k)>0$ for $k\in \begin{bmatrix}
0, N-1
\end{bmatrix}$, then the ENE policy for a nominal non-optimal solution is modified as
\begin{equation}
  \begin{aligned}
    \label{law2}
    &\delta u(k) = K^{*}_1(k) \delta x(k) + K^{*}_2(k) \delta w(k)\\
    &\hspace{12 mm} + K^{*}_3(k) \begin{bmatrix}
    f^T_{u}(k) T(k+1) + {H}_u(k) \\
    0
    \end{bmatrix},\\
    &K^{*}_3(k) = -
    \begin{bmatrix}
    I & 0
    \end{bmatrix}
    K^o(k),
  \end{aligned}
\end{equation}
where the gain matrices $K^{*}_1$, $K^{*}_2$, and $K^o(k)$ are defined in \eqref{K0} and \eqref{law}, and $T(k)$ is a non-zero variable defined in \eqref{T_k-large}.
\end{theorem}

\begin{proof}
Using \eqref{NE-KKT} and \eqref{NE-Hamiltonian-large}, (23) is modified as 
\begin{equation}
  \begin{aligned}
    \label{H_u-large}
    &H_{ux}(k) \delta x(k) + H_{uu}(k) \delta u(k) + H_{u w}(k) \delta w(k) \\
    &+ f^T_{u}(k) \delta \lambda(k+1) + {C^a_{u}}^T(k) \delta \mu(k) + H_{u}(k) = 0.
      \end{aligned}
\end{equation}
Substituting the linearized model \eqref{NE} and \eqref{lambda_form} into \eqref{H_u-large} yields
\begin{equation}
  \begin{aligned}
    \label{Z_uu-large}
    &Z_{ux}(k) \delta x(k) + Z_{uu}(k) \delta u(k) + Z_{u w}(k) \delta w(k) \\
    &+ {C^a_{u}}^T(k) \delta \mu(k) + f^T_{u}(k) T(k+1) + H_{u}(k) = 0.
      \end{aligned}
\end{equation}
Using the linearized safety constraints \eqref{NE} and \eqref{Z_uu-large}, one can obtain
\begin{equation}
  \begin{aligned}
    \label{u&mu-large}
    &\begin{bmatrix}
    \delta u(k)\\
    \delta \mu(k)
    \end{bmatrix}
    = -K^o(k)
    \begin{bmatrix}
    Z_{ux}(k) \\
    C^a_{x}(k)
    \end{bmatrix}
    \delta x(k) \\
    & \hspace{17 mm} -K^o(k)
    \begin{bmatrix}
    Z_{u w}(k) \\
    C^a_{w}(k)
    \end{bmatrix}
    \delta w(k) \\
    & \hspace{17 mm} -K^o(k)
    \begin{bmatrix}
    f^T_{u}(k) T(k+1) \\
    0
    \end{bmatrix}\\
    & \hspace{17 mm} -K^o(k)
    \begin{bmatrix}
    H_{u}(k) \\
    0
    \end{bmatrix}.
  \end{aligned}
\end{equation}
Substituting \eqref{u&mu-large} into \eqref{lambda1} yields
\begin{equation}
\small
  \begin{aligned}
    \label{lambda2-large}
    &\delta \lambda(k) = \left( Z_{xx}(k) -
    \begin{bmatrix}
    Z_{xu}(k) \hspace{1 mm} {C^{a}_{x}}^T(k)
    \end{bmatrix}
    K^o(k)
    \begin{bmatrix}
    Z_{ux}(k)\\
    C^{a}_{x}(k)
    \end{bmatrix} \right) \delta x(k) \\
    &+ \left( Z_{x w}(k) -
    \begin{bmatrix}
    Z_{xu}(k) \hspace{1 mm} {C^{a}_{x}}^T(k)
    \end{bmatrix}
    K^o(k)
    \begin{bmatrix}
    Z_{u w}(k)\\
    C^a_{w}(k)
    \end{bmatrix} \right) \delta w(k) \\
    &+ f^T_{x}(k) T(k+1) -
    \begin{bmatrix}
    Z_{xu}(k) \hspace{1 mm} {C^{a}_{x}}^T(k)
    \end{bmatrix}
    K^o(k)
    \begin{bmatrix}
    f^T_{u}(k) T(k+1)\\
    0
    \end{bmatrix}\\
    &+ g^T_{x}(k) \bar{T}(k+1) -
    \begin{bmatrix}
    Z_{xu}(k) \hspace{1 mm} {C^{a}_{x}}^T(k)
    \end{bmatrix}
    K^o(k)
    \begin{bmatrix}
    H_{u}(k)\\
    0
    \end{bmatrix}.
      \end{aligned}
\end{equation}
From \eqref{S_k}, \eqref{lambda_form}, and \eqref{lambda2-large}, it follows that
\begin{equation}
  \begin{aligned}
    \label{T_k-large}
    &T(k) = g^T_{x}(k) \bar{T}(k+1) + f^T_{x}(k) T(k+1) \\
    & - \begin{bmatrix}
    Z_{xu}(k) \hspace{1 mm} {C^{a}_{x}}^T(k)
    \end{bmatrix}
    K^o(k)
    \begin{bmatrix}
    f^T_{u}(k) T(k+1) + H_{u}(k)\\
    0
    \end{bmatrix}.
      \end{aligned}
\end{equation}
Now, plugging \eqref{u&mu-large} into \eqref{lambda_bar1} yields
\begin{equation}
\small
  \begin{aligned}
    \label{lambda_bar2-large}
    &\delta\bar{\lambda}(k) = \left( Z_{wx}(k) -
    \begin{bmatrix}
    Z_{wu}(k) \hspace{1 mm} {C^{a}_{w}}^T(k)
    \end{bmatrix}
    K^o(k)
    \begin{bmatrix}
    Z_{ux}(k)\\
    C^{a}_{x}(k)
    \end{bmatrix} \right) \delta x(k) \\
    &+ \left( Z_{w w}(k) -
    \begin{bmatrix}
    Z_{wu}(k) \hspace{1 mm} {C^{a}_{w}}^T(k)
    \end{bmatrix}
    K^o(k)
    \begin{bmatrix}
    Z_{u w}(k)\\
    C^a_{w}(k)
    \end{bmatrix} \right) \delta w(k) \\
    &+ f^T_{w}(k) T(k+1) -
    \begin{bmatrix}
    Z_{wu}(k) \hspace{1 mm} {C^{a}_{w}}^T(k)
    \end{bmatrix}
    K^o(k)
    \begin{bmatrix}
    f^T_{u}(k) T(k+1)\\
    0
    \end{bmatrix}\\
    &+ g^T_{w}(k) \bar{T}(k+1) -
    \begin{bmatrix}
    Z_{wu}(k) \hspace{1 mm} {C^{a}_{w}}^T(k)
    \end{bmatrix}
    K^o(k)
    \begin{bmatrix}
    H_{u}(k)\\
    0
    \end{bmatrix}.
      \end{aligned}
\end{equation}
Using \eqref{S_k_bar}, \eqref{lambda_form_bar}, and \eqref{lambda_bar2-large}, one has
\begin{equation}
  \begin{aligned}
    \label{T_k_bar-large}
    &\bar{T}(k) = g^T_{w}(k) \bar{T}(k+1) + f^T_{w}(k) T(k+1) \\
    & - \begin{bmatrix}
    Z_{wu}(k) \hspace{1 mm} {C^{a}_{w}}^T(k)
    \end{bmatrix}
    K^o(k)
    \begin{bmatrix}
    f^T_{u}(k) T(k+1) + H_{u}(k)\\
    0
    \end{bmatrix}.
      \end{aligned}
\end{equation}
Based on \eqref{u&mu-large}, \eqref{T_k-large} and \eqref{T_k_bar-large}, the modified ENE policy \eqref{law2} is obtained. This completes the proof.
\end{proof}

Now, to deal with large perturbations that may change the sets of inactive and active constraints, the perturbed values of $C(x(k),u(k),w(k))$ and $\mu (k)$ are analyzed to determine the inactive and active constraints under the perturbations. Using \eqref{u&mu-large}, the relation between the state and preview perturbations and the Lagrange multiplier perturbation is expressed as
\begin{equation}
  \begin{aligned}
    \label{mu perturbation}
    &\delta \mu(k) = K^{*}_4(k) \delta x(k) + K^{*}_5(k) \delta w(k),\\
    &K^{*}_4(k) = -
    \begin{bmatrix}
    0 & I
    \end{bmatrix}
    K^o(k)
    \begin{bmatrix}
    Z_{ux}(k)\\
    {C^a_{x}}(k)
    \end{bmatrix},\\
    &K^{*}_5(k) = -
    \begin{bmatrix}
    0 & I
    \end{bmatrix}
    K^o(k)
    \begin{bmatrix}
    Z_{u w}(k)\\
    {C^a_{w}}(k)
    \end{bmatrix}.
  \end{aligned}
\end{equation}

Moreover, using \eqref{law2}, the constraint perturbation is represented as
\begin{equation}
  \begin{aligned}
    \label{constraint perturbation}
& \delta C(k) = C_x(k) \delta x(k) + C_u(k) \delta u(k) + C_w(k) \delta w(k)\\
&\hspace{9.25 mm} = (C_x(k)+C_u(k)K^{*}_1(k)) \delta x(k)\\
&\hspace{13 mm}+ (C_w(k)+C_u(k)K^{*}_2(k)) \delta w(k)\\
&\hspace{13 mm}+ C_u(k)K^{*}_3(k) (f^T_{u}(k) T(k+1) + {H}_u(k)).
  \end{aligned}
\end{equation}

The perturbed Lagrange multiplier and the perturbed constraint are given by
\begin{equation}
  \begin{aligned}
    \label{perturbed mu}
& \mu(k) = \mu^o(k) + \delta \mu(k),
  \end{aligned}
\end{equation}
\begin{equation}
  \begin{aligned}
    \label{perturbed constraint}
& C(k) = C^o(k) + \delta C(k).
  \end{aligned}
\end{equation}

Different activity statuses of the constraints may occur due to large perturbations. To address this issue, we consider a line that connects the nominal variables $x^o(0)$ and $w^o(0)$ to the perturbed variables $x(0)$ and $w(0)$. For the connecting line, we identify several intermediate points such that the status of the constraint remains the same between two consecutive points. Since $\mu(k) = 0$ and $C(k) = 0$ for the inactive and active constraints, respectively, we use \eqref{perturbed mu} for the active constraints to find the intermediate points which make the constraints inactive. Specifically, for the active constraints, an $\alpha(k)$ ($0\leq\alpha (k)\leq1$) is computed to have $\mu^o(k) + \alpha (k) \delta \mu(k) = 0$. Moreover, we employ \eqref{perturbed constraint} for the inactive constraints to find the intermediate points which make the constraints active. For the inactive constraints, the $\alpha (k)$ is computed to have $C^o(k) + \alpha (k) \delta C(k) = 0$. Thus, for $k\in [0, N-1]$, the intermediate points are achieved using the following equation:
\begin{equation}
  \begin{aligned}
    \label{alpha}
    & \alpha(k) =\begin{cases} 
    -\frac{\mu^o(k)}{\delta \mu(k)} & \text{if} \hspace{2 mm} k \in \mathbb{K}^a \vspace{2 mm},\\
    -\frac{C^o(k)}{\delta C(k)} &  \text{if} \hspace{2 mm} k \in \mathbb{K}^i.
    \end{cases}
  \end{aligned}
\end{equation}

The smallest $\alpha(k)$ is found such that the obtained perturbation changes the activity statuses of the constraints at least at one time step $k$.

\begin{algorithm}[ht]
    \caption{Modified Extended Neighboring Extremal}
    \label{Large ENE}
    \textbf{Input}: The functions $f$, $g$, $C$, $\phi$, and $\psi$, and the nominal trajectory $\mathbf{x}^{o}(0:N)$, $\mathbf{u}^{o}(0:N)$, and $\mathbf{w}^{o}(0:N)$.\\
    \textbf{1}: Set $j=0$.\\
    \textbf{2}: Initialize the matrices $\lambda^o (N)$, $\bar{\lambda}^o(N)$, $S(N)$, $W(N)$, $\bar{S}(N)$, and $\bar{W}(N)$ using \eqref{Optimal KKT}, \eqref{S_N} , and \eqref{S_N_bar}, respectively.\\
    \textbf{3}: Calculate, in a backward run, the Lagrange multipliers $\mu^o(k)$, $\lambda^o(k)$, and $\bar{\lambda}^o(k)$ using \eqref{Lagrange multipliers}.\\ 
    \textbf{4}: Calculate, in a backward run, the matrices $Z(k)$, the gains $K^*_1(k)$, $K^*_2(k)$, $K^*_3(k)$, $K^*_4(k)$, and $K^*_5(k)$, and the matrices $S(k)$, $W(k)$, $T(k)$, $\bar{S}(k)$, $\bar{W}(k)$, and $\bar{T}(k)$ using \eqref{Zuu}, \eqref{law2}, \eqref{mu perturbation}, \eqref{S_k}, \eqref{S_k_bar}, \eqref{T_k-large}, and \eqref{T_k_bar-large}, respectively.\\
    \textbf{5}: Given initial state variation $\delta x(0)$ and initial preview variation $\delta w(0)$, in a forward run, calculate $\delta \mu(k)$, $\delta C(k)$, $\alpha (k)$, $\delta x(k+1)$, and $\delta w(k+1)$ using \eqref{mu perturbation}, \eqref{constraint perturbation}, \eqref{alpha}, \eqref{law2}, and the variations of the system \eqref{NE}, respectively.\\
    \textbf{6}: Set, in a forward run, $\alpha(k) = 1$ if $\alpha(k) < 0$ or $\alpha(k) > 1$. Then, find $\lambda =\text{min}(\alpha(k))$. If $\lambda=0$, change the activity status of the constraint for the corresponding time step $k$ and go to Step 2. \\
    \textbf{7}: Given $x^o(0)$, $w^o(0)$, $\lambda \delta x(0)$, and $\lambda \delta w(0)$, in a forward run, calculate $\delta u(k)$, $u(k)$, $\delta x(k+1)$, $\delta w(k+1)$, $x(k+1)$, and $w(k+1)$ using \eqref{law2} and the variations of the system \eqref{NE}.\\
    \textbf{8}: If $0<\lambda<1$, set $x^o(0) = x^o(0) + \alpha \delta x(0)$, $w^o(0) = w^o(0) + \alpha \delta w(0)$, $\delta x(0)=(1-\alpha)\delta x(0)$, $\delta w(0)=(1-\alpha)\delta w(0)$, and $j=j+1$. Then, go to Step 2.\\
    \textbf{9}: If $\lambda=1$, in a forward run, calculate $u^*(k) = u^o(k)+\sum^{}_{j} \delta u_j(k)$, $x(k+1)$, and $w(k+1)$ using \eqref{law2} and \eqref{system}. 
\end{algorithm}

Algorithm 2 summarizes the modified ENE procedure for adaptation the pre-computed nominal non-optimal control solution to the large state and preview perturbations such that it achieves the optimal control as $u^{*}(k) = u^{o}(k) + \delta u(k)$ using Theorem \ref{theo2}. The algorithm identifies the intermediate points and determines the modified ENE adaptation policy.

\begin{remark} [Designing Parameters]
\label{Designing Parameters}
Considering suitable nominal models, the main design parameters of the proposed approach come from the original optimization problem \eqref{NMPC}, which are the prediction number $N$ and the designing weights in the stage cost $\phi(x,u,w)$ and the terminal cost $\psi(x,w)$. The prediction number $N$ must be high enough so that the obtained optimal controller stabilizes the system; however, higher $N$ causes higher computational cost to solve the optimization problem. Moreover, the designing weights in the costs must be selected such that both minimum tracking error and minimum control input are achieved. 
\end{remark}

\begin{remark} [Implementation]
\label{Implementation}
The proposed ENE framework is easy to implement and light in computation. Specifically, given a nominal initial state ${x}^o(0)$, a nominal preview ${w}^o(0:N)$, a control objective function to minimize, system and control constraints, a nominal optimal state and control trajectory ${x}^o,{u}^o$ will be computed using an optimal control strategy. Note that this nominal solution can be computed offline and stored online, can be performed on a remote powerful controller (e..g, cloud), or computed ahead of time by utilizing the idling time of the processor. In the same time, the ENE adaptation gains $K^{*}_1(k),\,K^{*}_2(k),k=0,1,\cdots,N-1$ in \eqref{law} can also be computed along with the nominal control law. During the online implementation, the actual initial state $x(0)$ and the actual preview $w$ are likely different from the nominal values used for the optimal control computations. Instead of recomputing the optimal control sequence, the control correction \eqref{law} is computed, where $\delta x(k) = x(k) - x^{o}(k)$ and $\delta w(k) = w(k) - w^{o}(k)$ denote the state perturbation and the preview perturbation, respectively. Then the final control is used as $u^{*}(k) = u^{o}(k) + \delta u(k)$. This implementation is easily extended for the modified ENE. As seen from the steps discussed above, the proposed approach is easy to implement and involves negligible online computational cost.
\end{remark}

\begin{remark} [Nonlinear Model Predictive Control]
\label{Nonlinear Model Predictive Control}
One can employ the nonlinear optimal control problem \eqref{NMPC} as the open-loop nonlinear model predictive control (NMPC) or the closed-loop NMPC. For the open-loop version, providing the $N$-length nominal trajectory from the NMPC, the ENE algorithm approximates the NMPC policy such that it calculates two time-varying $N$-length feedback gains on the state and preview perturbations. Although the feedback gains are pre-computed, the ENE is able to take feedback from the real system for the $N$ predictions in contrast to the open-loop NMPC. On the other hand, for the colsed-loop NMPC, we save the ENE solution but we only apply the first control input to the plant at each time step. Taking the feedback from the real system, the ENE solution from the previous step is considered as the nominal non-optimal solution, and the ENE algorithm is applied again to adapt the recent solution for the current time step.
\end{remark}

\begin{remark} [Comparison]
\label{Comparison}
In comparison with the existing NE frameworks \cite{bagherzadeh2023neighboring,ghaemi2019optimal,gupta2017combined,bloch2016neighboring}, we extend the regular NE approaches that only consider state deviations to a general setting that both state and preview deviations are considered. This is a significant extension as many modern control applications are employing preview information due to the increased availability of connectivity \cite{Hajida1,amini2019cabin,laks2011model,yazdandoost2022optimization}. The necessity of adapting to preview perturbations is also demonstrated in our simulation studies, where we show that the proposed ENE can significantly outperform the regular NE when the preview information has certain variations. 
\end{remark}

\section{Simulation Results}
\label{Sec5}
In this part, we demonstrate the performance of the proposed ENE framework for both small and large perturbations via a simulation example. 
The simulation example is adopted from the cart-inverted pendulum (see Fig. \ref{Cart-inverted pendulum fig}) whose system dynamics is described by: 
\begin{equation}
  \begin{aligned}
    \label{Cart-inverted pendulum}
    &\ddot{z}=\frac{F-{K}_{d}\dot{z}-m(L{\dot{\theta }}^{2} \sin (\theta )-g\sin (\theta ) \cos (\theta ))-2w_z}{M+m \hspace{0.5 mm} {\sin }^{2}(\theta )},\\ 
    &\ddot{\theta }=\frac{\ddot{z}\cos (\theta )+g\sin (\theta )}{L}-\frac{w_{\theta}}{mL^2},
  \end{aligned}
\end{equation}
where $z$ and $\theta$ denote the position of the cart and the pendulum angle. $m=1$kg, $M=5$kg, and $L=2$m represent the mass of the pendulum, the mass of the cart, and the length of the pendulum, respectively. $g=9.81$m/s$^{2}$ and ${K}_{d}=10$Ns/m are respectively the gravity acceleration and the damping parameter. The variable force $F$ controls the system under a friction force $w_z$ and a friction torque $w_{\theta}$. $T=0.1$s is considered as the sampling time for discretization of the model \eqref{Cart-inverted pendulum}, and we assume that we have certain preview of $w_z$ and $w_{\theta}$.
The states, the outputs, the preview information, and the control input constraint are respectively expressed as
\begin{equation*}
  \begin{aligned}
    \label{pendulum's states}
    	x={[{x}_{1},{x}_{2},{x}_{3},{x}_{4}]}^{T}={[z,\dot{z},\theta,\dot{\theta }]}^{T},
  \end{aligned}
\end{equation*}
\begin{equation*}
  \begin{aligned}
    \label{pendulum's output}
    	y={[{x}_{1},{x}_{3}]}^{T}={[z,\theta]}^{T},
  \end{aligned}
\end{equation*}
\begin{equation*}
  \begin{aligned}
    \label{pendulum's preview}
    	w={[{w}_{1},{w}_{2},{w}_{3},{w}_{4}]}^{T}={[0,w_z,0,w_{\theta}]}^{T},
  \end{aligned}
\end{equation*}
\begin{equation*}
  \begin{aligned}
    \label{pendulum's constraints}
    	-300\le F\le 300.
  \end{aligned}
\end{equation*}

\begin{figure}[!ht]     
 \centering
     \includegraphics[width=0.65\linewidth]{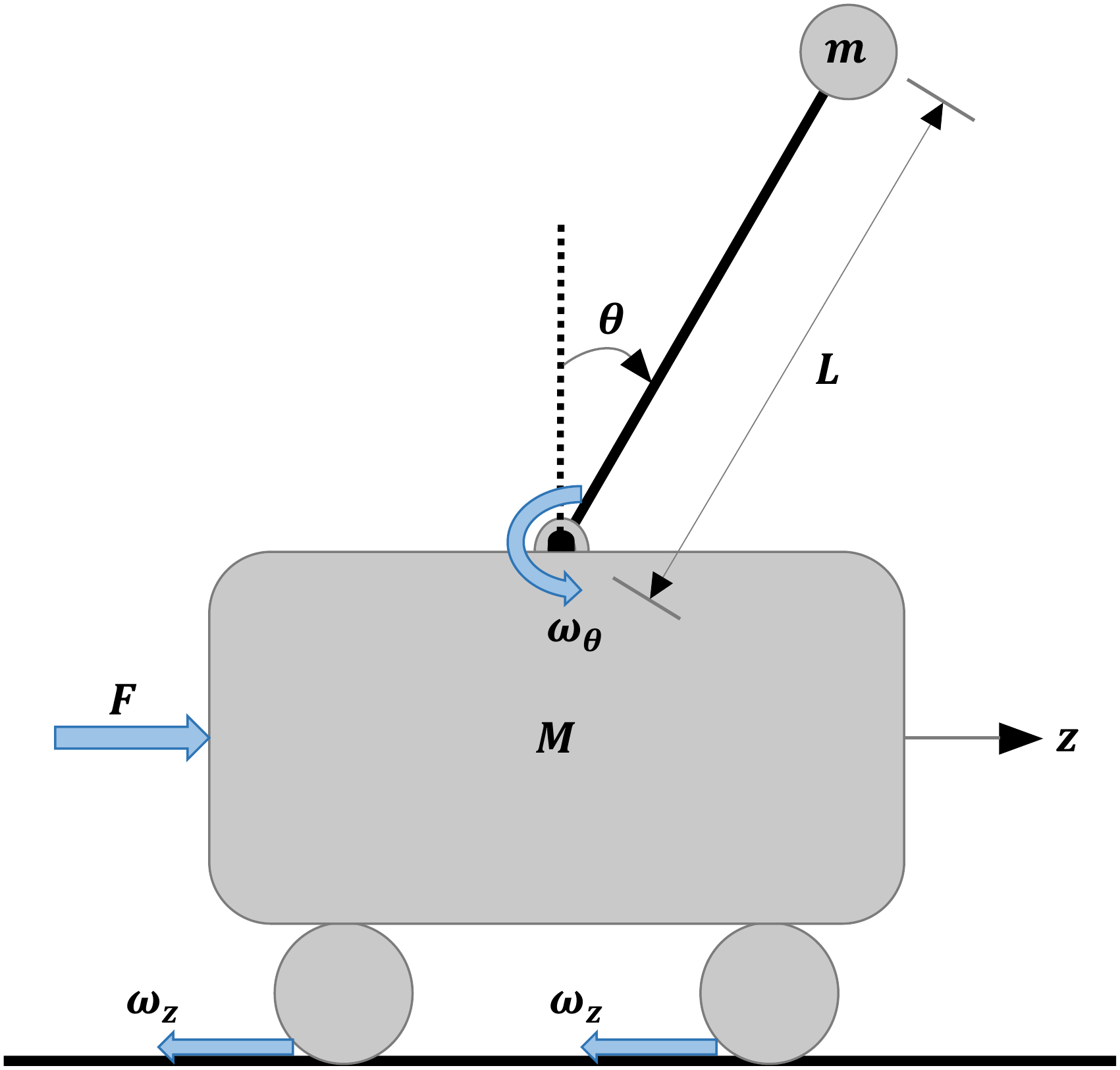}
     \caption{Cart-inverted pendulum.}
     \label{Cart-inverted pendulum fig}
\end{figure}

The following values are used for the simulation: $N = 35$, $x^o(0) = [0,0,-\pi,0]^T$, $w^o(0) = [0,0.1,0,0.1]^T$. Moreover, the nominal preview model is represented as $w^o(k+1) = -0.008 x^o(k) -0.1 w^o(k)$. For the small perturbation setting, the initial state perturbation and the actual friction profile are set as $\delta x (0) = [0.01,0.01,0.01,0.01]^T$ and $w(k)=0.004 \sin(k)+0.004 \text{rand}(k)+0.002$, respectively. For the large perturbation setting, the initial state perturbation and the actual friction profile are chosen as $\delta x (0) = [0.2,0.2,0.2,0.2]^T$ and $w(k)=0.015\sin(k)+0.015\text{rand}(k)+0.01$, respectively. 

Figs. \ref{Input small}-\ref{Cost small} show the control performance of the open-loop NMPC, the standard NE, the ENE, and the closed-loop NMPC subject to the small perturbations. For the open-loop NMPC, under the nominal initial state $x^o(0)$ and preview $w^o(0)$, we obtain the N-length open-loop trajectory $(x^{o},u^o,w^{o})$ and apply the open-loop control $u^o$ to the system as shown in Fig. \ref{Input small}. It is worth noting that the state and preview information are updated during the optimization problem based on the considered nominal model \eqref{Cart-inverted pendulum} and the nominal preview model $w^o(k+1) = -0.008 x^o(k) -0.1 w^o(k)$, respectively. However, since it is the open-loop version of the NMPC, the controller does not take the feedback from the real states and preview, makes the least control force in Fig. \ref{Input small}, and leads to degraded performance due to the state and preview deviations as shown in Fig. \ref{Outputs small}. The NE is capable of taking the state feedback from the real system and adjusting the nominal optimal control, the open-loop control trajectory obtained by the NMPC, for the state perturbations. From Fig. \ref{Outputs small}, one can see that the NE does show an improved performance as compared to the open-loop NMPC but it falls short against the ENE since it only handles the state perturbations without adapting to the preview perturbations. In comparison with the open-loop NMPC and the NE, the proposed ENE takes the state and preview feedback from the real system and achieves better performance, where it promptly stabilizes the system with the minimum cost in the presence of state and preview perturbations as shown in Fig. \ref{Cost small}. Although we employ the ENE for the open-loop NMPC, due to the feedback from the real system, the ENE shows a similar control performance as the closed-loop NMPC for this case as shown in Figs. \ref{Outputs small} and \ref{Cost small}. However, the closed-loop NMPC has high computational cost since it solves the optimization problem \eqref{NMPC} at each step.

\begin{figure}[!ht]
     \centering
     \includegraphics[width=0.99\linewidth]{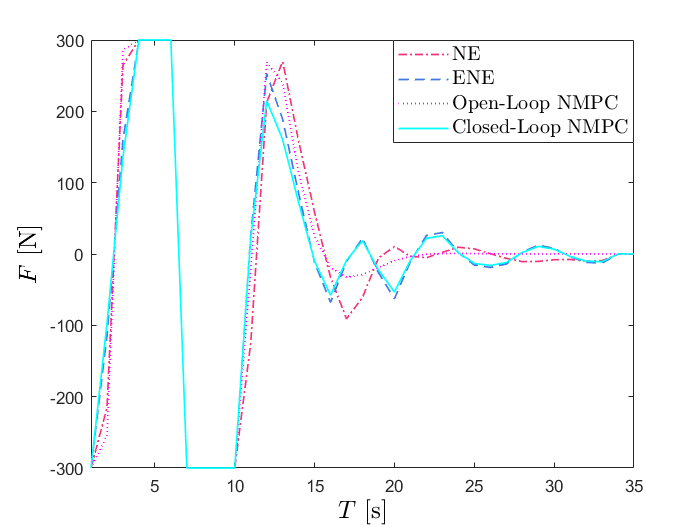}
    \caption{Control Input for Small Perturbation.}
     \label{Input small}
 \end{figure}
 
  \begin{figure}[!ht]
     \centering
     \includegraphics[width=0.99\linewidth]{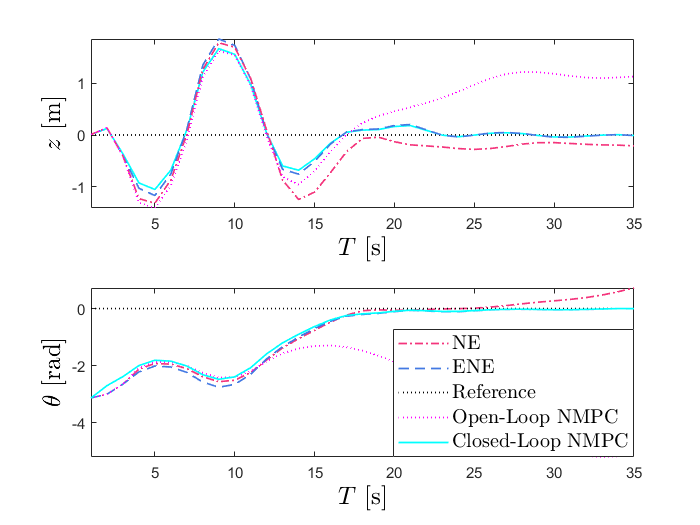}
    \caption{System Outputs for Small Perturbation.}
     \label{Outputs small}
 \end{figure}
 
  \begin{figure}[!ht]
     \centering
     \includegraphics[width=0.99\linewidth]{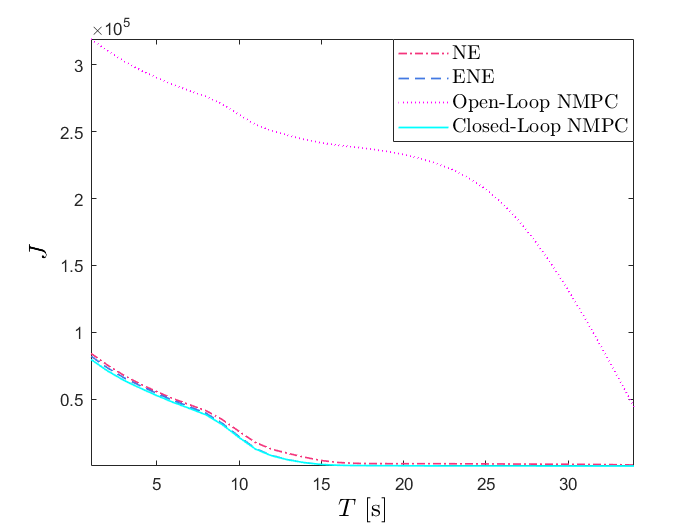}
    \caption{Cost for Small Perturbation.}
     \label{Cost small}
 \end{figure}

Figs. \ref{Input large} and \ref{Outputs large} illustrate the control performance of the open-loop NMPC, the NE, the ENE, the modified NE, the modified ENE, and the closed-loop NMPC subject to large perturbations. As shown in Fig. \ref{Input large}, one can see that the considered large perturbations change the activity status of the input constraint, and it causes that the NE and the ENE violates the constraint due to the absence of the intermediate points between the nominal initial state and preview and the perturbed ones. However, the modified NE and the modified ENE satisfies the constraint, and the modified ENE indicates a similar performance as the closed-loop NMPC as shown in Fig. \ref{Outputs large}. Moreover, to see the role of the nominal preview model on the proposed control scheme, Figs. \ref{Input comp} and \ref{Outputs comp} compare the results of the ENE and the modified ENE for two nominal preview models $w^o(k+1) = w^o(k)$ and $w^o(k+1) = -0.008 x^o(k) -0.1 w^o(k)$ with the actual friction profile $w(k)=0.008\sin(k)+0.008\text{rand}(k)+0.004$. One can see that the activity status of the constraint is changed under the considered perturbation; however, it is not high enough to cause the constraint violation for the ENE. Furthermore, it can be seen that both the modified ENE and the ENE accomplish better control performance when the preview model $w^o(k+1) = w^o(k)$ is applied. Providing a suitable nominal preview model leads to well control performance by the proposed ENE and modified ENE.

Table I compares the performances (i.e. $\lVert y - r \rVert$) and the computational times of the proposed controllers for the small perturbation. Based on the formulations, it is obvious that the ENE and the modified ENE (MENE) show the same performance and computational time for the small perturbations. We also have same result for the NE and the modified NE (MNE) for the small perturbations. Table II compares the performances and the computational times of the proposed controllers for the large perturbations. In Tables I and II, the closed-loop NMPC (CLNMPC) and the open-loop NMPC (OLNMPC) show the best and the worst performance, respectively; however, considering both performance and computational time, the modified ENE presents the best results.

The simulation setup is widely applicable as in many modern applications, a nominal preview model is available while the actual corresponding signal can also be measured or estimated online. For example, a wind energy forecast model is obtained using a deep federated learning approach \cite{ahmadi2022deep}, which can be served as a nominal preview model, and the wind disturbance can also be measured using light detection and ranging systems in real time \cite{laks2011model}. For the considered cart-inverted pendulum simulations, the nominal preview information is obtained using a nominal model, i.e. $w^o(k+1) = -0.008 x^o(k) -0.1 w^o(k)$; however, for each time step $k$, we generate the real preview information as $w(k) = 0.004 \sin(k) + 0.004 \text{rand}(k) + 0.002$, which leads to a perturbation from the nominal one. Providing a nominal solution based on the nominal state and preview, the proposed ENE framework adapts the nominal control to the perturbations generated by the measured/estimated real state and preview information. Furthermore, to simulate the large perturbation case, we follow the same process but change the real preview information as $w(k) = 0.015 \sin(k) + 0.015 \text{rand}(k) + 0.01$ for Figs. 6 and 7 and $w(k) = 0.008 \sin(k) + 0.008 \text{rand}(k) + 0.004$ for Figs. 8 and 9.
 
  \begin{figure}[!ht]
     \centering
     \includegraphics[width=0.99\linewidth]{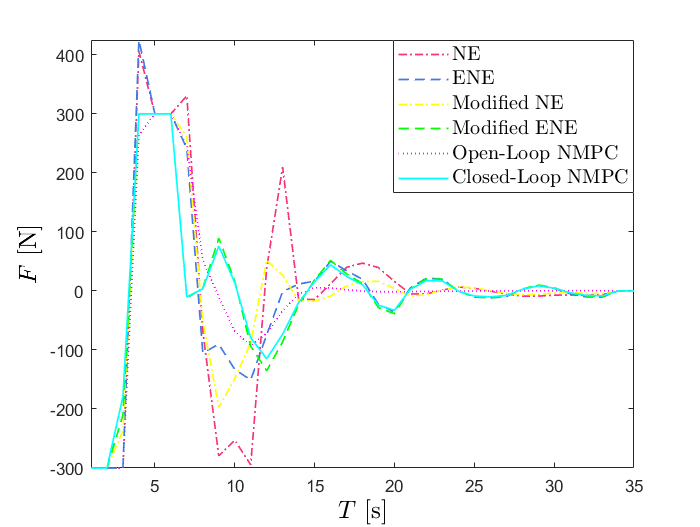}
    \caption{Control Input for Large Perturbation.}
     \label{Input large}
 \end{figure}
 
  \begin{figure}[!ht]
     \centering
     \includegraphics[width=0.99\linewidth]{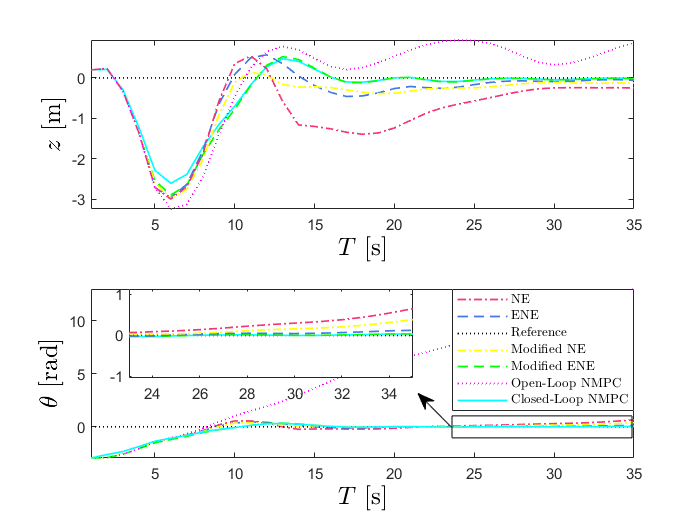}
    \caption{System Outputs for Large Perturbation.}
     \label{Outputs large}
 \end{figure}
 
   \begin{figure}[!ht]
     \centering
     \includegraphics[width=0.99\linewidth]{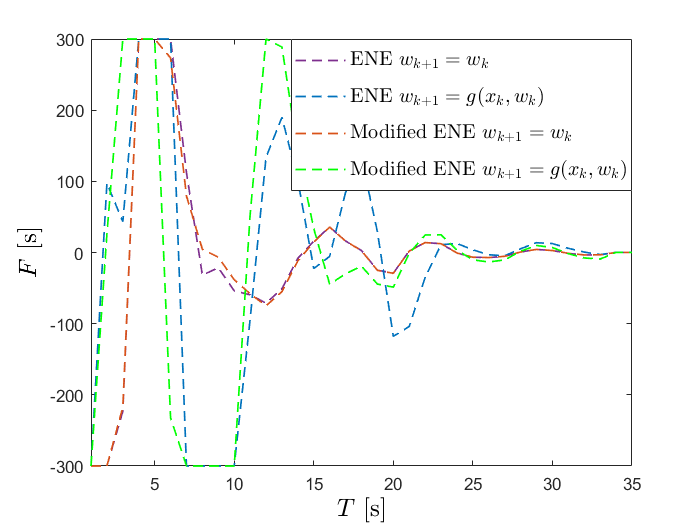}
    \caption{Control Input for Different Nominal Preview Models.}
     \label{Input comp}
 \end{figure}
 
   \begin{figure}[!ht]
     \centering
     \includegraphics[width=0.99\linewidth]{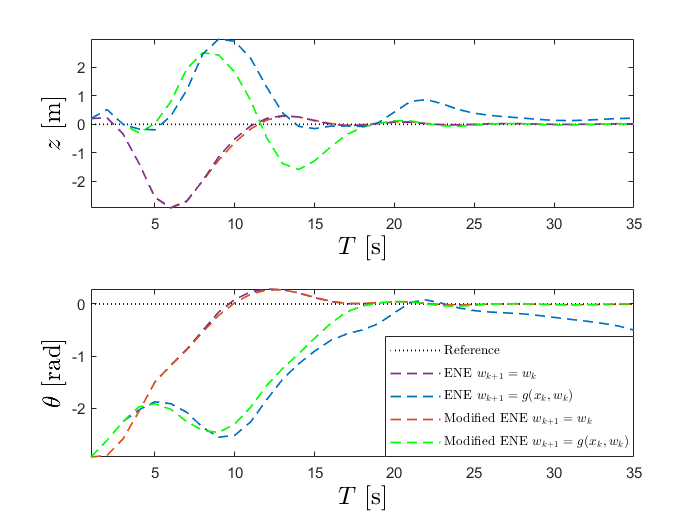}
    \caption{System Outputs for Different Nominal Preview Models.}
     \label{Outputs comp}
 \end{figure}  

\begin{table}[!ht]
\centering
 \caption{Comparison of Controllers for Small Perturbations}
\begin{tabular}{ |p{2.3cm}|p{2.3cm}|p{2.3cm}|  }
\hline\hline
Control & Performance & Time (per loop) \\
\hline
CLNMPC & $5.5735$ & $5.7179 \hspace{1 mm} ms$ \\\hline
MENE & $5.6429$ & $0.0659 \hspace{1 mm} ms$ \\\hline
ENE & $5.6429$ & $0.0659 \hspace{1 mm} ms$ \\\hline
MNE & $5.9846$ & $0.0658 \hspace{1 mm} ms$ \\\hline
NE & $5.9846$ & $0.0658 \hspace{1 mm} ms$ \\\hline
OLNMPC & $19.2123$ & $0.1770 \hspace{1 mm} ms$ \\\hline
\hline
\end{tabular}
\end{table}

\begin{table}[!ht]
\centering
 \caption{Comparison of Controllers for Large Perturbations}
\begin{tabular}{ |p{2.3cm}|p{2.3cm}|p{2.3cm}|  }
\hline\hline
Control & Performance & Time (per loop) \\
\hline
CLNMPC & $6.1609$ & $5.7179 \hspace{1 mm} ms$ \\\hline
MENE & $6.2704$ & $0.1225 \hspace{1 mm} ms$ \\\hline
ENE & $6.6838$ & $0.0659 \hspace{1 mm} ms$ \\\hline
MNE & $6.7045$ & $0.1224 \hspace{1 mm} ms$ \\\hline
NE & $7.4038$ & $0.0658 \hspace{1 mm} ms$ \\\hline
OLNMPC & $41.8378$ & $0.1770 \hspace{1 mm} ms$ \\\hline
\hline
\end{tabular}
\end{table}

\section{Conclusion}
\label{Sec6}
In this work, an ENE algorithm was developed to approximate the nonlinear optimal control policy for the modern applications which incorporate the preview information. The developed ENE was based on the second-order variation of the original optimization problem, which led to a set of Riccati-like backward recursive equations. The ENE adapted a nominal trajectory to the state and preview perturbations, and a multi-segment strategy was employed to guarantee well closed-loop performance and constraint satisfaction for the large perturbations. Simulations of the cart inverted pendulum system demonstrated the ENE's technological advances over the NE and the NMPC. Additionally, it was demonstrated that the nominal preview model is crucial to the effectiveness of the ENE. The proposed ENE framework is applicable to general optimal control problem setting as there is no assumption on the under/over-actuation of the system. If a regular optimal control implementation can yield good performance, the ENE is expected to yield comparable performance with less computation complexity. The computational load of the ENE grows linearly for the optimization horizon, which alleviates the online computational burden and extends the applicability of the optimal controllers. The main contribution of this paper is mainly on the proposed new framework with control law derivations and analysis. The main purpose of the simulation is to demonstrate the effectiveness of the proposed framework by showing that the ENE is able to achieve improved performance (as compared to the open-loop NMPC and the standard NE) with negligible online computation (as compared to the closed-loop NMPC). The considered cart-inverted pendulum is a classical system frequently used for the nonlinear control benchmarks \cite{desouky2022lyapunov}. In our future work, we will evaluate the developed ENE framework on real-world physical systems such as robots and autonomous vehicles. Furthermore, we will consider the reference perturbation for tracking control problems and also develop a data-enabled ENE to remove the requirement of having an explicit system model.


\bibliographystyle{ieeetr}
\bibliography{References.bib}
 
\end{document}